\documentclass{llncs}
\usepackage{amsmath,amssymb}
\begin{document}
\newcommand{\WC}{\widehat{C}}
\newcommand{\WD}{\widehat{D}}
\newcommand{\WT}{\widehat{T}}
\newcommand{\WS}{\widehat{S}}

\title{Recognizing Read-Once Functions from Depth-Three Formulas}
%
%
\author{Alexander Kozachinskiy}
%
%
%
\institute{
National Research University Higher School of Economics\\ \email{akozachinskiy@hse.ru}
}

\maketitle              

\begin{abstract}
Consider the following decision problem: for a given  monotone Boolean function $f$ decide, whether $f$ is read-once.
  For this problem, it is essential how the input
  function $f$ is represented.
  On a negative side we have the following results. Elbassioni, Makino and
  Rauf (\cite{elbassioni2011readability}) proved that this problem is
  coNP-complete when $f$ is given by a depth-4 read-2 monotone Boolean formula.
  Gurvich (\cite{gurvich2010it}) proved that this problem is coNP-complete
  even when the input is the following expression: $C\lor D_n$,
  where $D_n = x_1 y_1 \lor \ldots  \lor x_n y_n$ and $C$ is a
  monotone CNF over the variables $x_1, y_1, \ldots, x_n, y_n$ (note that
  this expression is a monotone Boolean formula of depth 3;
  in \cite{gurvich2010it} nothing is said about the readability of $C$, but the proof is valid even if $C$ is read-2 and thus the entire formula is read-3). 

  On a positive side, from \cite{golumbic2009read} we know that there
  is a polynomial time algorithm to recognize read-once functions
  when the input is a monotone depth-2 formula (that is, a
  DNF or a CNF).
  There are even very fast algorithms for this problem
  (~\cite{golumbic2008improvement}).

  We show that we can test in polynomial-time whether a given
  expression $C\lor D$ computes a read-once function,
  provided that $C$ is a read-once monotone CNF and $D$ is a
  read-once monotone DNF and all the variables of $C$ occur also in $D$
  (recall that due to Gurvich, the problem is coNP-complete
  when $C$ is read-2). We also observe that from the so-called Sausage Lemma of Boros et al. (\cite{boros2009generating}) it follows that   the problem of recognizing read-once functions is coNP-complete when the input formula is depth-3 read-2. This  improves the result of \cite{elbassioni2011readability}
  in the depth and the result of \cite{gurvich2010it} in the
  readability of the input formula. Moreover, 
we prove a new variant of Sausage Lemma, which allows us to derive that read-once recognition is coNP-complete even for  depth-3 read-2 formulas of the form $\Psi\land D_n$, where $D_n$
  is  as above and $\Psi$ is a  $\bigwedge-\bigvee-\bigwedge$ depth-3
  read-once monotone Boolean formula.

\end{abstract}

\section{Introduction}

In this paper we study the following decision problem: decide, for a given monotone Boolean function $f$, whether $f$ is a read-once function (the latter
means that the function can be computed by a
monotone formula in which every
variable occurs only once). Of course, to specify the problem, we need to specify the representation of $f$. For some representations this problem turns out to be tractable. For example, it is known (see \cite{golumbic2009read}) that if $f$ is given by a monotone DNF (or, equivalently, CNF), then the corresponding problem can be solved in polynomial time. Golumbic, Mintz and Rotics (\cite{golumbic2008improvement}) gave quite fast algorithm for this problem which works  in time $O(nm)$, where $m$ is the length of the given DNF and $n$ is the number of variables. In  \cite{golumbic2008improvement} they also asked how hard is this problem when $f$ is represented in some other way, for example,  when $f$ is given by an arbitrary Boolean formula.

First of all, let us note that if the input is a monotone
Boolean formula, then the problem belongs to coNP.
This follows from the following theorem by Gurvich:
\begin{theorem}[\cite{gurvich1977repetition}]
\label{gurvich_criteria}
A monotone Boolean function $f$ is read-once iff every minterm $S$ of $f$ and every maxterm $T$ of $f$ intersect in exactly one point.
\end{theorem}
Thus to show that $f$ is \emph{not} read-once it is enough to demonstrate a minterm $S$ and a maxterm $T$ with $|S\cap T| > 1$ (it is not hard
to show that given a formula $S,T$, we can decide in polynomial time
whether $S$ is a minterm and $T$ is a maxterm).

Soon after  Golumbic, Mintz and Rotics raised their question, Elbassioni, Makino and Rauf (\cite{elbassioni2011readability}) proved that the read-once recognition problem is coNP-complete when the input function $f$ is given by a
depth-4 read-2 monotone Boolean formula. The same authors also proved that it is NP-hard to approximate the readability of the monotone Boolean function $f:\{0, 1\}^n\to \{0, 1\}$, given by a depth-4 monotone Boolean formula, within a factor of $O(n)$.  

Later, Gurvich (\cite{gurvich2010it}) proved, that the problem of recognizing read-once functions is coNP-complete even when the input is the following expression: $C\lor D_n$, where $D_n = x_1 y_1 \lor \ldots  \lor x_n y_n$ and $C$ is a monotone CNF over the variables $x_1, y_1, \ldots, x_n, y_n$. Note that the
entire formula $C\lor D_n$ is a depth-3 formula. The paper \cite{gurvich2010it} says nothing about the readability of $C$, but the proof is valid even if $C$ is  read-2 \footnote{This is due to the fact that the SAT-problem is NP-complete even for read-3 (non-monotone) CNFs.} (so that the whole expression is read-3).

Our main result shows that this problem becomes tractable,
if $C$ is read-once even when
 $D_n$ is any monotone read-once DNF:
\begin{theorem}
\label{upper_bound}
There is a polynomial-time algorithm which decides, whether
a given expression $C\lor D$ computes a read-once function,
provided that $C$ is a monotone read-once CNF, $D$ is a monotone read-once
DNF, and every variable of $C$ occurs also in $D$.
\end{theorem}

We don't know whether the last restriction can be removed.  However, we find this theorem interesting due to its connection to the result of Gurvich.

We also observe that from the so-called Sausage Lemma of Boros et al. (\cite{boros2009generating}) it follows that   the problem of recognizing read-once functions is coNP-complete when the input formula is depth-3 read-2 (note that  \cite{elbassioni2011readability} requires depth at least 4 and  \cite{gurvich2010it} requires readability at least 3 for the input formula). Moreover, we may consider only formulas which are conjunctions of two read-once formulas.

In more detail, Sausage Lemma states that it is coNP-complete problem to decide whether $\Psi\to D$ is a tautology, where $\Psi$ is a $\bigwedge-\bigvee-\bigwedge$ depth-3 read-once monotone Boolean formula and $D$ is a read-once monotone DNF. Further, it is easy to show (using Theorem \ref{gurvich_criteria}) that $\Psi\to D$ is a tautology if and only if the following formula: $$\Psi \land (w_1 w_3 \lor w_2 w_4) \land (D \lor w_1 w_2 \lor w_3 w_4)$$ computes a read-once function. Here $w_1, w_2, w_3, w_4$ are fresh variables.

Using a completely different reduction, we prove that Sausage Lemma is true even when $D$ is of the same form as in the Gurvich's hardness result from  \cite{gurvich2010it}.

\begin{theorem}
\label{sausage_lemma}
The problem to decide whether $\Psi\to D_n$ is a tautology is coNP-complete. Here  
$D_n = x_1 y_1 \lor \ldots \lor x_n y_n$ and $\Psi$ is a  $\bigwedge-\bigvee-\bigwedge$ depth-3 read-once monotone Boolean formula over $x_1, y_1, \ldots, x_n, y_n$.
\end{theorem}

The same trick with  fresh variables allows us to derive the following
\begin{corollary}
\label{lower_bound}

It is a coNP-complete problem to decide whether a given
expression $\Psi\land D_n$ computes a read-once function. Here  
$D_n = x_1 y_1 \lor \ldots \lor x_n y_n$ and $\Psi$ is a  $\bigwedge-\bigvee-\bigwedge$ depth-3 read-once monotone Boolean formula over $x_1, y_1, \ldots, x_n, y_n$.
\end{corollary}

\textbf{Remark.} Corollary \ref{lower_bound} can be used to show that inapproximability result of Elbassioni, Makino and Raur is also true for depth-3 formulas. This is, however, can be done with the use of the result of Gurvich as well.

\section{Preliminaries}
A monotone Boolean formula (i.e., a $\land,\lor$-formula) $\Phi$ is called a
\emph{read-$k$ formula} if every variable occurs at most $k$ times in $\Phi$. A monotone Boolean function $f$ is called a \emph{read-$k$ function}
if there is a monotone  read-$k$ formula, computing $f$.
\emph{Readability} of a Boolean function $f$ (formula $\Phi$)
is the minimal $k$ such that function $f$ (formula $\Phi$) is read-$k$.

Assume that $f$ is a monotone Boolean function over the variables $x_1, \ldots, x_n$ and $S$ is a subset of $\{x_1, \ldots, x_n\}$. To simplify notation below let $f(S\to i)$ (here $i\in\{0, 1\}$) denote the value of $f$ when all the variables from $S$ are set to $i$ and all the variables from $\{x_1, \ldots, x_n\}\setminus S$ are set to $1 - i$.

A subset $S\subset \{x_1, \ldots, x_n\}$ is called a \emph{minterm} of $f$ if $f(S\to 1) = 1$ but for every proper subset $S^\prime$ of $S$ it holds that $f(S^\prime\to 1) = 0$. Similarly, a subset  $T\subset\{x_1, \ldots, x_n\}$ is called a \emph{maxterm} of $f$ if $f(T\to 0) = 0$ but for every proper subset $T^\prime$ of $T$ it holds that $f(T^\prime\to 0) = 1$.

Obviously, every minterm of $f$ intersects every $f$'s maxterm.

\section{Proof of Theorem \ref{upper_bound}}

Our algorithm uses the following lemma:
\begin{lemma}
\label{polynomial_time_lemma}
There exists a polynomial-time algorithm which for any
given read-once monotone CNF $C$ and
for any given read-once monotone DNF $D$ decides,
whether $C\to D$ is a tautology.
\end{lemma}
\begin{proof}
It is known (see, e.g., \cite{buning1999propositional})  that there is a polynomial-time algorithm to decide, whether a given read-2 CNF is satisfiable. Apply this algorithm to $\lnot(C\to D) =  C\land \lnot D$ (the latter can be re-written as a read-2 CNF in polynomial time).
\qed
\end{proof}

Let $\{x_1, \ldots, x_n\}$ be variables occurring in $D$.
Let $C_1, \ldots, C_m$ denote the clauses of $C$. Since $C$ is read-once, we may identify $C_1, \ldots, C_m$ with  $m$ disjoint subsets of $\{x_1,\ldots, x_n\}$.
The same thing can be done for $D$. Let
$$D = D_1 \lor D_2 \lor \ldots \lor D_l,$$
where $D_1, \ldots, D_l \subset \{x_1, \ldots, x_n\}$ are disjoint conjunctions. Note also that $D_1\cup\ldots \cup D_l = \{x_1, \ldots, x_n\}$. 

We provide first a description of minterms of $C\lor D$. Let $S$ be
a subset of $\{x_1, \ldots, x_n\}$. We say that  $S$ is
a \emph{right set} if for some $j\in\{1, \ldots, l\}$ we have $S = D_j$.
We say that $S$ is a \emph{left set} if  $S\subset C_1\cup \ldots \cup C_m$
and for every $i\in\{1, \ldots, m\}$ it holds that $|C_i\cap S| = 1$.
The following lemma is straightforward.
\begin{lemma}
\label{left_right_lemma}
A set $S$ is a minterm of $C\lor D$ if and only if
$S$ is a left set
that does not properly include any right set or
$S$ is a right set
that does not properly include any left  set.
\end{lemma}

A minterm $S$ of $C\lor D$ is called \emph{a left minterm} if $S$ is
a left set. Similarly, we call $S$ \emph{a right minterm}
if $S$ is a right set.

Now we are ready to present the algorithm for Theorem \ref{upper_bound}.
In the description of the algorithm we will state several auxiliary lemmas
whose proofs are deferred to Appendix.

\medskip
\textbf{Algorithm.} The algorithm works in four steps.

\emph{Step 1.} Check, using Lemma \ref{polynomial_time_lemma}, whether $C\to D$ is a tautology. If it is, then $C\lor D$ is equivalent to $D$ and hence
$C\lor D$ computes a read-once function and the algorithm halts.
Otherwise proceed to Step 2.

\emph{Step 2.}
Obviously every maxterm $T$
includes at least one clause of $C$.
For every pair of distinct clauses $C_u, C_v$ check
whether there is a maxterm $T$ of $C\lor D$ such that $C_u, C_v\subset T$.
This can be done in  polynomial time by 
the following
\begin{lemma}
\label{two_clauses_lemma}
Let $C_u, C_v$ be two distinct clauses of $C$. Then there exists a maxterm $T$ of $C\lor D$ such that $C_u, C_v \subset T$ if and only if for every $j\in\{1, \ldots, l\}$ it holds that $|(C_u\cup C_v) \cap D_j| \le 1$.
\end{lemma}

If there is such $T$, then  $C\lor D$ is not read-once.
Indeed, consider any 
minimal $S_0 \subset\{x_1, \ldots, x_n\}$ such that $C(S_0 \to 1) = 1$
and $D(S_0 \to 1) = 0$. Such a set exists, since $C\to D$ is not a tautology. 
As $D(S_0 \to 1) = 0$, the set $S_0$ does not include
any right set, and hence is a left minterm of $C\lor D$.
As, $C(S_0 \to 1) = 1$,
the set $S_0$ intersects both  $C_u$ and $C_v$. Recall that
$C_u, C_v\subset T$ and hence $|T\cap S_0|\ge2$. By 
Theorem~\ref{gurvich_criteria}, 
$C\lor D$ does not compute a read-once function.

Otherwise (if there is no
maxterm $T$ of $C\lor D$ that includes distinct clauses of $C$)
we proceed to Step 3.

\emph{Step 3.}
For every clause $C_u$ and for every pair of distinct variables $p$
and $q$ from $C_u$ we check:
\begin{itemize}
\item whether there is a right minterm $S$ such that $\{p, q\} \subset S$
  (this can be done in polynomial time since there are only polynomially many right minterms);
\item whether there is a maxterm $T$ containing  $C_u$.
\end{itemize}

The second check can be done in polynomial time
using the following 
\begin{lemma}
\label{max_term_lemma_1}
Assume that $C\to D$ is not a tautology and no maxterm of $C\lor D$ contains two distinct clauses of $C$.
Then for any clause $C_i$
we can decide in polynomial-time  
whether there exists a maxterm $T$ of $C\lor D$ such that $C_i\subset T$. 
\end{lemma}

If for some $C_u, p, q$ both questions answer in positive,
then $C\lor D$ is not a read-once function.
Indeed, in this case both the minterm $S$ and the maxterm $T$
include distinct variables $p$
and $q$, and  by 
Theorem~\ref{gurvich_criteria}, 
$C\lor D$ does not compute a read-once function.
Otherwise we proceed to Step 4.

\emph{Step 4.} For all $u\in \{1, \ldots, m\}$ and
$v\in\{1, \ldots, l\}$ with
$C_u \cap D_v = \varnothing$, and for all $p\in D_v$ and $q\in C_u$
we check:

\begin{itemize}
\item whether there is a left minterm $S$ such that $\{p, q\} \subset S$,
\item whether there is a maxterm $T$ such that $C_u\cup \{p\} \subset T$.
  \end{itemize}
Both checks can be performed in polynomial time
by the following lemmas.
\begin{lemma}
\label{minterm_lemma}
There exists a polynomial-time algorithm which for any
given pair of distinct variables $a, b\in \{x_1, \ldots, x_n\}$ decides, whether there is a left minterm $S$ of $C\lor D$ such that $\{a, b\} \subset S$.
\end{lemma}

\begin{lemma}
  \label{max_term_lemma_2}
  Assume that
  there is no maxterm of $C\lor D$ which contains two distinct clauses of $C$.
  Then for any  given $C_u,D_v$ with $C_u \cap D_v = \varnothing$
  and for any given
  $p \in D_v$ we can decide in polynomial time
  whether there exists a maxterm $T$ of $C$ such that
  $C_u\cup \{p\} \subset T$. \footnote{The proof of this lemma is almost identical to the proof of Lemma \ref{max_term_lemma_1}. Actually, it is possible to formulate a single lemma which implies both of them, but then the formulation of the lemma becomes immense.}
\end{lemma}

If for some $C_u, D_v, p, q$ both questions answer in positive,
then $C\lor D$ does not compute a read-once function.
Indeed,  in this case both the maxterm $T$ and the minterm $S$
include distinct variables $p,q$ and by 
Theorem~\ref{gurvich_criteria}, $C\lor D$ does not compute a read-once function.

Otherwise the algorithm
outputs the positive  answer 
and halts. 
We have to show that in this case
$C\lor D$ indeed computes a read-once function. 
For the sake of contradiction, assume that $C\lor D$ is not read-once.
By Theorem \ref{gurvich_criteria}
there is a maxterm $T$ of $C\lor D$ and a minterm
$S$ of  $C\lor D$ that have distinct common variables $p,q$.
By Lemma \ref{left_right_lemma} $S$ is either a
left or a right minterm. We will consider these two cases
separately.

\emph{Case 1:} $S$ is a right set, say $S = D_j$. 
Then  $T$ contains some clause $C_u$ of $C$.
We claim that $S \cap T\subset C_u$.
Indeed, otherwise $S\cap T$ would include a variable $y\notin C_u$. 
Then $C(T\setminus \{y\} \to 0) = 0$, as
$C_u \subset T\setminus \{y\}$.
And $D(T\setminus \{y\} \to 0) = 0$, since
$D$ is read-once and $T\cap D_j$ has at least 2 variables.
We obtain  contradiction, as $T$ is a maxterm of $C\lor D$.

Thus $T\cap S \subset C_u$.
Hence there are two distinct variables $p$ and $q$ from
$C_u$ such that $\{p, q\} \subset S $. Hence 
the algorithm must have halted on Step 3,
a contradiction.

\emph{Case 2:} $S$ is a left set.
Again there is a clause $C_u\subset T$.
Since $S$ is a left set, $S$ and $C_u$ have exactly one common
variable $q$.
By our assumption $T\cap S$ includes another variable
$p\ne q$. Note that $p\notin C_u$ because otherwise $S\cap C_u$
has more than one variable.

Let  $D_v$ be the (unique) right set 
that includes $p$ (here we use the assumption
that every variable is in some right set).
We claim that $D_v$ and $C_u$ are disjoint.
For the sake of contradiction assume that there is a variable
$y\in C_u \cap D_v$.
Then $(C\lor D)(T\setminus \{p\}\to 0) = 0$.
Indeed,
$C(T\setminus \{p\} \to 0) = 0$, as $p\notin C_u$ and hence $C_u\subset T\setminus \{p\}$. And $D(T\setminus \{p\}\to 0) = 0$, as
$y\ne p$ (since $y$ is in $C_u$ and $p$ is not) and hence   
$T\setminus \{p\}$ still intersects $D_v$.

Thus $T$ is not a maxterm, and the contradiction shows
that  $D_v$ and $C_u$ are disjoint.
Therefore, we have found $C_u, p, q, D_v$ such that $q\in C_u$, $p\in D_v$, $C_u$ and $D_v$ are disjoint and there is a left minterm $S$ and a maxterm $T$ with $\{p, q\} \subset S, C_u\cap \{p\} \subset T$.
Hence the algorithm just have halted on Step 4, a contradiction.
\textbf{(End of Algorithm.)}

\section{Lower bound}

We start this section with the proof of our variant of the Sausage Lemma (Theorem \ref{sausage_lemma}).

\begin{proof}[of Theorem \ref{sausage_lemma}]
We reduce from a clique problem. Assume that we are given a simple undirected graph $G = (V, E)$ and an  integer $k$ and we want to know whether there is a clique of size $k$ in $G$. The reduction is divided into two steps. On step 1 we introduce an auxiliary game between 3 players, Alice, Bob and Merlin. In this game Alice and Bob cooperatively play against Merlin. The rules of the game depend on $G$, $k$. The structure of this game resembles so-called GKS communication games (see \cite{gilmer2015new}). We will show that there is a clique of size $k$ in $G$ if and only if Alice and Bob have a winning strategy in a game. On step 2 of the reduction we will construct (in polynomial time) a $\bigwedge-\bigvee-\bigwedge$ depth-3 read-once monotone Boolean formula $\Psi$ with the following feature: Alice and Bob have a winning strategy in a game if and only if $\Psi\to D_n$ is not a tautology. The value of $n$ will  depend on $k$ and on the size of $G$. 

\emph{The game.} There is a room with a tape consisting of $k$ blank cells. The game has two phases. During the first phase Alice and Merlin are in the room and Bob is outside. He does not see what happens in a room. Alice and Merlin interact according to the following protocol. At the beginning Merlin points out to one of $k$ cells. Then Alice writes some vertex $u$ of $G$ in this cell. Finally, Merlin  writes some vertex $v$ of $G$ in some other cell. The only restriction is that $v$ should not be connected with $u$ by an edge of $G$ (for example, $v$ can be equal to $u$, we assume that there are no self-loops in $G$).  The first phase is finished.

Then, during the second phases, Bob comes in and Alice with Merlin are outside. Bob sees that exactly two cells of a tape are not blank --- each of these two cells contains a vertex of a graph. To win, Bob should determine which cell was touched first. More precisely, the cell which was touched first is the one to which Merlin pointed out in the beginning and in which Alice wrote her vertex.

More formally, let $\mathrm{CONF}$ denote the following set:
$$\mathrm{CONF} = \{\{(i, u), (j, v)\} : i, j\in\{1, \ldots, k\}, u, v\in V, i\neq j, \{u, v\} \notin E \},$$

Note that $\mathrm{CONF}$ encodes all possible configurations Bob can see when he comes in. 

A strategy of Alice in this game  is a function $S_A: \{1, 2, \ldots, k\} \to V$, a strategy of Bob is a function $S_B: \mathrm{CONF} \to \{1, 2, \ldots, k\}$
and a strategy of Merlin is a triple $(i, S_M^1, S_M^2)$, where $i\in\{1, 2, \ldots, k\}$ and 
$$S_M^1: V \to \{1, 2, \ldots, k\}, \qquad S_M^2: V \to V.$$
This triple should satisfy the following condition: for every $u\in V$ it holds that $S_M^1(u) \neq i$ and $\{u, S_M^2(u)\} \notin E$. We call a pair $(S_A, S_B)$  a winning strategy, if for every Merlin's strategy $(i, S_M^1, S_M^2)$ it holds that
$$S_B(\{\,\,(i, S_A(i)),\,\, (S_M^1(S_A(i)), S_M^2(S_A(i)))\,\,\}) = i.$$

Let us show that if there is clique of size $k$ in $G$, then Alice and Bob have a winning strategy in the above game. Indeed, assume that $w_1, \ldots, w_k\in V$ form a $k$-clique in $G$.  Alice's strategy is the following: if Merlin points out to the $i^{th}$ cell, then she writes $w_i$ in it. Assume that after that Merlin picks $j^{th}$ cell and writes something in it.  How can Bob distinguish the $i^{th}$ cell from the $j^{th}$ cell? The point is that for $v$ which is written by Merlin in the $j^{th}$ cell we have that $v \neq w_j$, because $w_i$ and $w_j$ are connected by an edge. This allows Bob to win.

On the other hand, assume that Alice and Bob have a winning strategy in a game. Let $v_i$ be the vertex which Alice writes according to this strategy when Merlin points out to the $i^{th}$ cell. Let us show that $v_1, \ldots, v_k$ form a clique in $G$. For contradiction assume that there is $i\neq j$ such that $v_i$ and $v_j$ are not connected by edge. Note that there are two  strategies of Merlin resulting into the same (from the Bob's point of view) configuration,  namely:
\begin{itemize}
\item Merlin points to the  $i^{th}$ cell and then writes $v_j$ in the $j^{th}$ cell;
\item Merlin points to the $j^{th}$ cell and then writes $v_i$ in the $i^{th}$ cell.
\end{itemize}
Both these strategies are legal since $\{v_i, v_j\} \notin E$. This contradicts the fact that from this configuration Bob can determine the cell which was touched first.

Now let us proceed to the second step of a reduction. There are exactly $2 \cdot |\mathrm{CONF}|$ outcomes of Alice-Merlin interaction (for each $c\in\mathsf{CONF}$ there are two outcomes resulting in $c$). Introduce a Boolean variable for each possible outcome.  Namely, let $x^{i, u}_{j, v}$ correspond to an outcome in which Merlin picked the $i^{th}$ cell, Alice  wrote $u$ in the $i^{th}$ cell and then Merlin wrote $v$ in the $j^{th}$ cell.

Every assignment to these variables, on which the value of  $D_n$ equals 0, determines Bob's strategy. Here
$D_n = \bigvee_{\{(i, u),  (j, v)\}\, \in \, \mathrm{CONF}} x^{i, u}_{j, v} \land x^{j, v}_{i, u}$. 
Namely, for $c = \{(i, u), (j, v)\}$ define:
$$S_B(c) = \begin{cases} i & \mbox{if $x^{i, u}_{j, v} = 1$,} \\ j & \mbox{if $x^{j, v}_{i, u} = 1$.} \end{cases}$$
For technical reasons, if both $x^{i, u}_{j, v}$ and $x^{j, v}_{i, u}$ are set to 0, define $S_B(c)$ somehow in such a way that $S_B(c) \neq i, j$. Note that it never happens that both $x^{i, u}_{j, v}$ and $x^{j, v}_{i, u}$ are set to 1 (otherwise $S_B(c)$ would be undetermined for $c = \{(i, u), (j, v)\}$). It is also clear that every Bob's strategy can be encoded in this way by some assignment of variables on which $D_n = 0$.

Once  a Bob's strategy (or equivalently, an assignment to $x^{i, u}_{j, v}$) is fixed, the game is just between Alice and Merlin.  Alice wins if and only if an outcome of her interaction with Merlin corresponds to a variable which is set to 1.   Indeed, if an outcome  corresponds to $x^{i, u}_{j, v}$, then Bob  should output $i$ and it happens if and only if $x^{i, u}_{j, v}$ is set to 1 (here in $(\Longrightarrow)$-direction we use the fact that if both  $x^{i, u}_{j, v}$ and $x^{j, v}_{i, u}$ are set to 0, then $S_B(c) \neq i, j$). This means that Alice has a winning strategy if and only if $\Psi = 1$, where
$$\Psi =\ \bigwedge_{i = 1}^k \bigvee_{u\in V} \bigwedge_{\substack{(j, v):\\ \{(i, u), (j, v)\}\in\mathrm{CONF}}} x^{i, u}_{j, v}$$
($\bigwedge$-gates correspond to Merlin's moves and $\bigvee$-gate corresponds to Alice's move). 
In other words, Alice and Bob have a winning strategy if and only if there is an assignment to $x^{i, u}_{j, v}$ such that $D_n$ is false and $\Psi$ is true or, equivalently, $\Psi\to D_n$ is not a tautology.
\qed
\end{proof}

To derive Corollary \ref{lower_bound} we need the following simple Lemma.

\begin{lemma}
\label{simple_lemma}
A function $$f(w_1, w_2, w_3, w_4) = w_2 w_3 w_4 \lor w_1 w_3 w_4 \lor w_1 w_2 w_4 \lor w_1 w_2 w_3$$
is not read-once.
\end{lemma}
\begin{proof}
Note that $\{w_1, w_2, w_3\}$ is a minterm of $f$ and $\{w_1, w_2\}$ is a maxterm of $f$. Hence by Theorem \ref{gurvich_criteria} $f$ is not read-once. 
\qed
\end{proof}

\begin{proof}[of Corollary \ref{lower_bound}]
Assume that $\Psi$ is a $\bigwedge-\bigvee-\bigwedge$ depth-3 read-once monotone Boolean formula over $x_1, y_1, \ldots, x_n, y_n$ and $D_n = \bigvee_{i = 1}^n x_i y_i$. It is enough to show that $\Psi\to D_n$ is a tautology if and only if $$(\Psi \land (w_1 w_3 \lor w_2 w_4) ) \land (D_n \lor w_1 w_2 \lor w_3 w_4)$$ computes a read-once function. Indeed, if $\Psi\to D_n$ is a tautology, then $$(\Psi \land (w_1 w_3 \lor w_2 w_4) ) \to  (D_n \lor w_1 w_2 \lor w_3 w_4)$$ is also a tautology. Hence
$$ (\Psi \land (w_1 w_3 \lor w_2 w_4) ) \land (D_n \lor w_1 w_2 \lor w_3 w_4) = \Psi \land (w_1 w_3 \lor w_2 w_4),$$
and the latter is a read-once formula.  On the other hand assume for contradiction that $\Psi\to D_n$ is not a tautology, but there is a read-once formula $\Phi(x_1, y_1, \ldots, x_n, y_n,  w_1, \ldots, w_4)$, which is logically equivalent to   $$(\Psi \land (w_1 w_3 \lor w_2 w_4) ) \land (D_n \lor w_1 w_2 \lor w_3 w_4).$$
Let $a_1, b_1, \ldots, a_n, b_n$ be an assignment to $x_1, y_1, \ldots, x_n, y_n$ such that $$\Psi(a_1, b_1 \ldots, a_n, b_n) = 1, \qquad D_n(a_1, b_1 \ldots, a_n, b_n) = 0.$$ Substitute $x_1\to a_1, y_1 \to b_1, \ldots x_n \to a_n, y_n \to b_n$ in $\Phi$. The resulting read-once formula will be
\begin{align*}
\Phi(a_1, b_1, \ldots, a_n, b_n, w_1, \ldots, w_4) &= (w_1 w_3 \lor w_2 w_4) \land (w_1 w_2 \lor w_3 w_4)\\
&= w_2 w_3 w_4 \lor w_1 w_3 w_4 \lor w_1 w_2 w_4 \lor w_1 w_2 w_3. 
\end{align*}
Thus we obtain a read-once formula for $w_2 w_3 w_4 \lor w_1 w_3 w_4 \lor w_1 w_2 w_4 \lor w_1 w_2 w_3$, this a contradiction with Lemma \ref{simple_lemma}.
\qed
\end{proof}

\subsubsection*{Acknowledgments}

The author would like to thank Alexander Shen and Nikolay Vereshchagin for help in writing this paper. The author would like to thank Vladimir Gurvich for pointing out to \cite{boros2009generating}.
%
%

\appendix

\section*{Appendix}

\section{Proof of Lemma~\ref{two_clauses_lemma}}
  Let  $T$ be a maxterm of $C\lor D$ and $C_u, C_v \subset T$.
  For the sake of contradiction assume that there is $j\in\{1, \ldots, l\}$
  such that $|(C_u\cup C_v) \cap D_j| \ge 2$. Pick any two distinct
  $p, q\in (C_u\cup C_v) \cap D_j$. Let us show that
  $(C\lor D)(T\setminus \{q\} \to 0) = 0$.
  To show that $C(T\setminus \{q\} \to 0) = 0$ observe that $C_u$ or $C_v$
  does not contain $q$ and hence $C_u\subset T\setminus \{q\}$ or $C_v \subset T\setminus \{q\}$. To show that $D(T\setminus \{q\} \to 0) = 0$ observe  that $D(T\to 0) = 0$ and hence $T$ intersects all sets
  $D_1, \ldots, D_l$. Since $q\in D_j$ and hence $q\notin D_i$
  for all $i\ne j$, the set  
  $T\setminus\{q\}$ still intersects $D_i$
  for all $i\ne j$. And it intersects  $D_j$ since
  $p\in C_u\cup C_v\subset T$ and $p$ was not removed from
  $T$. Since $T$ is a maxterm, this is  a contradiction. 

  On the other hand, assume that for every $j\in \{1, \ldots, l\}$ it holds
  that $|(C_u\cup C_v) \cap D_j| \le 1$. We have to find a maxterm $T$
  that includes  both $C_u$ and $C_v$. 
  Start with $T=C_u\cup C_v$. Then for all $j$ such that $D_j$ does not
  intersect $C_u\cup C_v$ pick a variable from $D_j$ and include it in $T$. 
  In this way we make $T$ intersect every $D_j$ in exactly one point.
  In particular,  $D(T\to 0) = 0$ and $C(T\to 0) = 0$. On the other hand,
  every proper subset $T^\prime$ of $T$ is disjoint with at least one $D_j$ and
  hence $D(T^\prime\to 0) = 1$. This shows that $T$ is a maxterm.

\section{Proof of Lemma \ref{max_term_lemma_1}}
Since $C\to D$ is not a tautology, $C_i$ is non-empty and intersects with some $D_j$. Further, without loss of generality we may assume that:
\begin{itemize}
\item $i = 1$;
\item $C_1$ intersects with $D_1, \ldots D_r$ and $C_1$ is disjoint with $D_{r + 1}, \ldots D_l$ for some $1 \le r \le l$;
\item $C_2, \ldots, C_s$ all intersect with $D_1 \cup D_2 \cup \ldots \cup D_r$ and $C_{s + 1}, \ldots, C_m$ are all disjoint with $D_1 \cup D_2 \cup \ldots \cup D_r$ for some $1 \le s \le m$.
\end{itemize}

From the fact that $D_1 \cup D_2 \cup \ldots \cup D_l = \{x_1, \ldots, x_n\}$ we may derive that:
\begin{equation}
\label{appendix_subset_2}
C_{s + 1}, \ldots, C_{m} \subset D_{r + 1} \cup \ldots \cup D_l.
\end{equation}

 Define an auxiliary CNF $\WC = C_{s + 1} \land \ldots \land C_{m}$ and an auxiliary DNF $\WD = D_{r + 1} \lor \ldots \lor D_l$. Note that $\WC$ and $\WD$ are over variables from $D_{r + 1} \cup \ldots \cup D_l$ (this follows from \eqref{appendix_subset_2}).

We claim that there exists $T$ such that $T$ is a maxterm of $C\lor D$ and $C_1\subset T$ if and only if $\WC \to \WD$ is not a tautology (the latter by Lemma \ref{polynomial_time_lemma} can be verified in polynomial time).

\textbf{($\Leftarrow$)}. Assume that $\WC\to \WD$ is not a tautology. Then there exists $\WT\subset D_{r + 1} \cup \ldots \cup D_l$ such that
\begin{equation}
\label{appendix_wt_1}
C_{s + 1} \not\subset \WT, \ldots, C_m \not\subset \WT;
\end{equation}
\begin{equation}
\label{appendix_wt_2}
|D_{r + 1} \cap \WT| = 1, \ldots, |D_{l} \cap \WT| = 1;
\end{equation}
(take minimal $\WT \subset D_{r + 1} \cup \ldots \cup D_l$ such that $\WC(\WT\to 0) = 1$ and $\WD(\WT\to 0) = 0$).

Let us show that $T = \WT \cup C_1$ is the maxterm of $C\lor D$. First of all, let us verify that $(C\lor D)(T\to 0) = 0$. Indeed,
\begin{itemize}
\item $C(T\to 0) = 0$ because $C_1 \subset T$;
\item $D(T\to 0) = 0$ because every $D_1, \ldots, D_l$ intersects with $T$; namely,  $D_1, \ldots D_r$ intersect $C_1$ and $D_{r + 1}, \ldots, D_l$ intersect $\WT$.
\end{itemize}

Now, assume that $T^\prime \subset T$ and $(C\lor D)(T^\prime \to 0) = 0$. Let us show that this is possible only when $T^\prime = T$.

Since $D(T^\prime \to 0) = 0$, we have that $T^\prime$ intersects with every $D_1, \ldots, D_l$. From the fact that $C_1$ is disjoint with $D_{r + 1}, \ldots, D_l$ and from \eqref{appendix_wt_2} it follows that $\WT \subset T^\prime$. 

 It remains to show that $C_1 \subset T^\prime$. This follows from the assumption that $C(T^\prime \to 0) = 0$. Indeed, then at least one clause of $C$ should be the subset of $T^\prime$. Assume that this clause is $C_u$. If $C_u \neq C_1$, then $C_u\subset \WT$. There are two cases:
\begin{itemize}
\item \emph{The first case.} Assume that $C_u \in \{C_2, \ldots, C_s\}$. Then $C_u\subset\WT \subset  D_{r + 1} \cup \ldots \cup D_l$ intersects with $D_1 \cup D_2 \cup \ldots \cup D_r$, but the latter is impossible.

\item \emph{The second case.} Assume that $C_u \in \{C_{s + 1}, \ldots, C_m\}$.  This case contradicts \eqref{appendix_wt_1}.
\end{itemize}

\textbf{($\Rightarrow$)} Assume that $T$ is the maxterm of $C\lor D$ such that $C_1 \subset T$. Define $\WT = T\setminus C_1$. Later we will show that $\WC(\WT\to 0) = 1$, $\WD(\WT\to 0) = 0$ and hence $\WC\to \WD$ is not a tautology. But at first we should verify that $\WT \subset D_{r + 1} \cup \ldots \cup D_l$ (recall that $\WC, \WD$ are over variables from $D_{r + 1} \cup \ldots \cup D_l$).

To show that $\WT \subset D_{r + 1} \cup \ldots \cup D_l$ assume for contradiction that $\WT$ intersects $D_1 \cup D_2  \cup \ldots \cup D_r$ and let $q$ be the variable which lies in their intersection. Note that $q\notin C_1$ (this is because $q\in \WT = T\setminus C_1$). Let us demonstrate that for such $q$ we have that $(C\lor D)(T\setminus \{q\} \to 0) = 0$ (this is already a contradiction since $T$ is a maxterm). Indeed, $C(T\setminus \{q\} \to 0) = 0$ since $C_1 \subset T\setminus \{q\}$. Further, we should show that $T\setminus \{q\}$ intersects every $D_1, \ldots D_l$ and hence $D(T\setminus \{q\} \to 0) = 0$. Indeed: 
\begin{itemize}
\item $T\setminus \{q\}$ intersects $D_1, \ldots, D_r$ because of $C_1$;
\item $T\setminus \{q\}$ intersects $D_{r + 1}, \ldots, D_{l}$ because of the following two reasons: (a) $D(T\to 0) = 0$ and hence $T$ intersects every $D_{r + 1}, \ldots, D_l$; (b) $q$ is not in $D_{r + 1}\cup\ldots \cup D_l$.
\end{itemize}

Thus it remains to show that $\WC(\WT\to 0) = 1$ and $\WD(\WT\to 0) = 0$. To show that the first equality is true  assume for contradiction that there is $C_u \in \{C_{s + 1}, \ldots, C_m\}$ such  that $C_u \subset \WT = T\setminus C_1$. But then $C_u \subset T$.  This contradicts the assumption that there is no maxterm of $C\lor D$ which contains two distinct clauses of $C$.

To show that the second equality ($\WD(\WT\to 0) = 0$) is true, observe that $\WT$ intersects every $D_{r + 1}, \ldots, D_l$. This is true because $D(\WT\to 0) = 0$ and hence $T$ intersects $D_{r + 1}, \ldots, D_l$; but $C_1$ by assumption is disjoint with $D_{r + 1}, \ldots, D_l$ and hence $T\setminus C_1$ still intersects them.

\section{Proof of Lemma~\ref{minterm_lemma}}

If there is $C_i$ such that $a, b\in C_i$ or $\{a, b\} \not\subset C_1 \cup\ldots \cup C_m$, then no left minterm $S$ can contain both $a$ and $b$. From now we assume that this is not the case, i.e. there is no $C_i$ which contains both $a$ and $b$  and $a, b\in C_1\cup \ldots \cup C_m$. 	Let $\WC \lor \WD$ be obtained from $C\lor D$ by setting $a, b$ to 1. In other words, $\WC$ is obtained from $C$ by erasing all clauses containing $a$ or $b$ and $\WD$ is obtained from $D$ by erasing $a$ and $b$. Assume without loss of generality that $C_1,  C_2$ are erased clauses, $a\in C_1, b\in C_2$ and $D_1, \ldots, D_r$ are conjunctions containing $a$ or $b$ (note that $r$ is either 1 or 2). Then $\WC$ and $\WD$ can be written as
$$\WC = C_{3} \land \ldots \land C_m,$$
$$\WD = (D_1\setminus \{a, b\}) \lor \ldots \lor (D_r\setminus \{a, b\}) \lor D_{r + 1} \lor \ldots \lor D_l.$$

 We assert that there is a left minterm $S$ containing $a$ and $b$ iff $\WC \to \WD$ is not a tautology. The latter by Lemma \ref{polynomial_time_lemma} can be verified in polynomial times.

\textbf{($\Leftarrow$)} Assume that $\WC \to \WD$ is not a tautology. Take minimal $\WS \subset \{x_1, \ldots, x_n\}\setminus \{a, b\}$ such that $\WC(\WS \to 1) = 1, \WD(\WS\to 1) = 0$. Obviously, such $\WS$ satisfies the following two conditions:
\begin{equation}
\label{min_t_1}
\WS\subset C_{3} \cup \ldots \cup C_m, \qquad |\WS \cap C_{3}| = 1, \ldots, |\WS \cap C_m| = 1,
\end{equation}

\begin{equation}
\label{min_t_2}
D_1\setminus \{a, b\} \not\subset \WS, \ldots, D_r\setminus \{a, b\} \not\subset \WS,\,\, D_{r + 1} \not\subset \WS, \ldots, D_l\not\subset \WS. 
\end{equation}

Now, define $S = \WS \cup \{a, b\}$. Let us show that $S$ is a left minterm of $C\lor D$. From \eqref{min_t_2} it follows that there is no $j \in \{1, \ldots, l\}$ such that $D_j\subset S$. Hence $S$ contains no right set as a proper subset. Thus it remains to show by Lemma \ref{left_right_lemma} that $S$ is a left set. Since $a, b$ are from $C_1 \cup \ldots \cup C_m$, we have that 
$$S\subset (\{a, b\} \cup C_{3} \cup \ldots \cup C_m) \subset C_1 \cup \ldots \cup C_m.$$

Moreover, $S$ intersects every clause of $C$ in exactly one point. For $C_{3}, \ldots, C_m$ this follows from \eqref{min_t_1} and from the fact that $C_{3}, \ldots, C_m$ contain neither $a$ nor $b$. For $C_1, C_2$ this is true because: (a) $\WS$ is disjoint with $C_1, C_2$; (b) $a\in C_1, b\in C_2$.

\textbf{($\Rightarrow$)} Assume that there is a left minterm $S$ of $C\lor D$ containing $a$ and $b$. Define $\WS = S\setminus \{a, b\}$. Let us show that $\WS$ intersects every $C_3, \ldots C_m$. Indeed, this is true for $S$ and $a, b$ are not from $C_3\cup \ldots \cup C_m$. Hence $\WC(\WS\to 1) = 1$. On the other hand, $\WD(\WS\to 1) = 0$, since:
\begin{itemize}
\item $D_1\setminus \{a, b\} \not \subset \WS, \ldots D_r\setminus \{a, b\}\not\subset \WS$ because otherwise at least one $D_1, \ldots, D_r$ is the subset of $S = \WS \cup \{a, b\}$;
\item $D_{r + 1}\not\subset \WS, \ldots, D_l \not\subset \WS$ because it is true even for $S$.
\end{itemize}
Thus $\WC\to \WD$ is not a tautology.

\section{Proof of Lemma \ref{max_term_lemma_2}}
Without loss of generality we may assume that:
\begin{itemize}
\item $i = j = 1$;
\item $C_1$ intersects with $D_2, \ldots D_r$ and $C_1$ is disjoint with $D_{r + 1}, \ldots D_l$ for some $1 \le r \le l$;
\item $C_2, \ldots, C_s$ all intersect with $D_1\setminus \{p\} \cup D_2 \cup \ldots \cup D_r$ and $C_{s + 1}, \ldots, C_m$ are all disjoint with $D_1\setminus \{p\} \cup D_2 \cup \ldots \cup D_r$ for some $1 \le s \le m$.
\end{itemize}

From the fact that $D_1 \cup D_2 \cup \ldots \cup D_l = \{x_1, \ldots, x_n\}$ we may derive that:
\begin{equation}
\label{subset_2}
C_{s + 1}, \ldots, C_{m} \subset \{p\} \cup D_{r + 1} \cup \ldots \cup D_l.
\end{equation}

Now, define:
$$\WC_{s + 1} = C_{s + 1}\setminus \{p\}, \ldots, \WC_{m} = C_{m}\setminus \{p\}.$$

Further, define an auxiliary CNF $\WC = \WC_{s + 1} \land \ldots \land \WC_{m}$ and an auxiliary DNF $\WD = D_{r + 1} \lor \ldots \lor D_l$. Note that $\WC$ and $\WD$ are over variables from $D_{r + 1} \cup \ldots \cup D_l$ (this follows from \eqref{subset_2}).

We claim that there exists $T$ such that $T$ is a maxterm of $C\lor D$ and $C_1\cup \{p\}\subset T$ if and only if $\WC \to \WD$ is not a tautology (the latter by Lemma \ref{polynomial_time_lemma} can be verified in polynomial time).

\textbf{($\Leftarrow$)}. Assume that $\WC\to \WD$ is not a tautology. Then there exists $\WT\subset D_{r + 1} \cup \ldots \cup D_l$ such that
\begin{equation}
\label{wt_1}
\WC_{s + 1} \not\subset \WT, \ldots, \WC_m \not\subset \WT;
\end{equation}
\begin{equation}
\label{wt_2}
|D_{r + 1} \cap \WT| = 1, \ldots, |D_{l} \cap \WT| = 1;
\end{equation}
(take minimal $\WT \subset D_{r + 1} \cup \ldots \cup D_l$ such that $\WC(\WT\to 0) = 1$ and $\WD(\WT\to 0) = 0$).

Let us show that $T = \WT \cup C_1 \cup \{p\}$ is a maxterm of $C\lor D$. First of all, let us verify that $(C\lor D)(T\to 0) = 0$. Indeed,
\begin{itemize}
\item $C(T\to 0) = 0$ because $C_1 \subset T$;
\item $D(T\to 0) = 0$ because every $D_1, \ldots, D_l$ intersects with $T$; namely, $D_1$ contains $p$, $D_2, \ldots D_r$ intersect $C_1$ and $D_{r + 1}, \ldots, D_l$ intersect $\WT$.
\end{itemize}

Now, assume that $T^\prime \subset T$ and $(C\lor D)(T^\prime \to 0) = 0$. Let us show that this is possible only when $T^\prime = T$.

Since $D(T^\prime \to 0) = 0$, we have that $T^\prime$ intersects with every $D_1, \ldots, D_l$. From the fact that $C_1 \cup \{p\}$ is disjoint with $D_{r + 1}, \ldots, D_l$ and from \eqref{wt_2} it follows that $\WT \subset T^\prime$. Since $C_1$ and $\WT$ are disjoint with $D_1$ and $T^\prime$ intersects $D_1$, we have that $p\in T^\prime$.

 It remains to show that $C_1 \subset T^\prime$. This follows from the assumption that $C(T^\prime \to 0) = 0$. Indeed, then at least one clause of $C$ should be the subset of $T^\prime$. Assume that this clause is $C_u$. If $C_u \neq C_1$, then $C_u\subset \WT \cup \{p\}$. There are two cases:
\begin{itemize}
\item \emph{The first case.} Assume that $C_u \in \{C_2, \ldots, C_s\}$. Then $C_u\subset\WT\cup \{p\} \subset \{p\} \cup D_{r + 1} \cup \ldots \cup D_l$ intersects with $D_1\setminus \{p\} \cup D_2 \cup \ldots \cup D_r$, but the latter is impossible.

\item \emph{The second case.} Assume that $C_u \in \{C_{s + 1}, \ldots, C_m\}$. Then $\WC_u = C_u\setminus \{p\} \subset \WT$, and the latter contradicts \eqref{wt_1}.
\end{itemize}

\textbf{($\Rightarrow$)} Assume that $T$ is the maxterm of $C\lor D$ such that $C_1\cup\{p\} \subset T$. Define $\WT = T\setminus (C_1 \cup \{p\})$. Later we will show that $\WC(\WT\to 0) = 1$, $\WD(\WT\to 0) = 0$ and hence $\WC\to \WD$ is not a tautology. But at first we should verify that $\WT \subset D_{r + 1} \cup \ldots \cup D_l$ (recall that $\WC, \WD$ are over variables from $D_{r + 1} \cup \ldots \cup D_l$).

To show that $\WT \subset D_{r + 1} \cup \ldots \cup D_l$ assume for contradiction that $\WT$ intersects $D_1 \cup D_2  \cup \ldots \cup D_r$ and let $q$ be the variable which lies in their intersection. Note that $q\neq p$ and $q\notin C_1$ (this is because $q\in \WT = T\setminus (C_1 \cup \{p\})$). Let us demonstrate that for such $q$ we have that $(C\lor D)(T\setminus \{q\} \to 0) = 0$ (this gives us a contradiction since $T$ is a maxterm). Indeed, $C(T\setminus \{q\} \to 0) = 0$ since $C_1 \subset T\setminus \{q\}$. Further, we should show that $T\setminus \{q\}$ intersects every $D_1, \ldots D_l$ and hence $D(T\setminus \{q\} \to 0) = 0$. Indeed: 
\begin{itemize}
\item $T\setminus \{q\}$ intersects $D_1$ because of $p$;
\item $T\setminus \{q\}$ intersects $D_2, \ldots, D_r$ because of $C_1$;
\item $T\setminus \{q\}$ intersects $D_{r + 1}, \ldots, D_{l}$ because of the following two reasons: (a) $D(T\to 0) = 0$ and hence $T$ intersects every $D_{r + 1}, \ldots, D_l$; (b) $q$ is not in $D_{r + 1}, \ldots, D_l$.
\end{itemize}

Thus it remains to show that $\WC(\WT\to 0) = 1$ and $\WD(\WT\to 0) = 0$. To show that the first equality is true  assume for contradiction that there is $\WC_u \in \{\WC_{s + 1}, \ldots, \WC_m\}$ such  that $\WC_u \subset \WT = T\setminus (C_1 \cup \{p\})$. But $\WC_u = C_u\setminus \{p\}$ and hence $C_u \subset T$.  This contradicts the assumption that there is no maxterm of $C\lor D$ which contains two distinct clauses of $C$.

To show that the second equality ($\WD(\WT\to 0) = 0$) is true, observe that $\WT$ intersects every $D_{r + 1}, \ldots, D_l$. This follows from the following two observations: (a) $D(T\to 0) = 0$ and hence $T$ intersects every $D_{r + 1}, \ldots D_l$, (b) $\{p\}\cup C_1$ is disjoint with $D_{r + 1}, \ldots, D_l$.

\end{document}